%% file: lipics-v2021-sample-article.tex
\documentclass[a4paper,UKenglish,cleveref,autoref,thm-restate]{lipics-v2021}

\usepackage{verbatim}
\usepackage{graphicx,color}
\def\eps{\epsilon}
\def\suchthat{\;:\;}
\def\given{\;|\;}

\def\err{\mathrm{err}}
\def\ind{\mathrm{ind}}

\newcommand{\abs}[1]{\left|#1\right|}

\newcommand{\norm}[1]{\left\|#1\right\|}

\newcommand{\prob}[1]{\operatorname{Pr}\left(#1\right)}

\newcommand{\size}[1]{\left|#1\right|}
\newcommand{\linspan}[1]{\operatorname{span}\left(#1\right)}

\newcommand{\expect}[2]{\underset{#1}{\operatorname{E}}\left[#2\right]}
\newcommand{\expectilde}[2]{\underset{#1}{\operatorname{\tilde{E}}}\left[#2\right]}

\usepackage[ruled,vlined]{algorithm2e}

\usepackage{bbm}
\usepackage{float}
\usepackage{soul}
\usepackage{amssymb,amsmath,amsthm}
\allowdisplaybreaks

\nolinenumbers

\usepackage{tikz,lmodern}
\usepackage[most]{tcolorbox}

\newcommand{\nnz}{\mathrm{nnz}}
\newcommand{\poly}{\mathrm{poly}}

\newcommand{\R}{\mathbb{R}}

\SetCommentSty{mycommfont}

\SetKwInput{KwInput}{Input}                
\SetKwInput{KwOutput}{Output}

\title{On Subspace Approximation and Subset Selection in Fewer Passes by MCMC Sampling} 

\titlerunning{} 

\author{Amit Deshpande}{Microsoft Research, Bangalore, India }{amitdesh@microsoft.com}{}{(Optional) }

\author{Rameshwar Pratap}{IIT Mandi, H.P., India}{rameshwar.pratap@gmail.com}{[orcid]}{[funding]}

\authorrunning{Deshpande and Pratap} 

\Copyright{Deshpande and Pratap} 



\category{} 

\relatedversion{} 

\acknowledgements{I want to thank \dots}

\EventEditors{John Q. Open and Joan R. Access}
\EventNoEds{2}
\EventLongTitle{42nd Conference on Very Important Topics (CVIT 2016)}
\EventShortTitle{CVIT 2016}
\EventAcronym{CVIT}
\EventYear{2016}
\EventDate{December 24--27, 2016}
\EventLocation{Little Whinging, United Kingdom}
\EventLogo{}
\SeriesVolume{42}
\ArticleNo{23}

\begin{document}

\maketitle

\begin{abstract}
We consider the problem of subset selection for $\ell_{p}$ subspace approximation, i.e., given $n$ points in $d$ dimensions, we need to pick a small, representative subset of the given points such that its span gives $(1+\eps)$ approximation to the best $k$-dimensional subspace that minimizes the sum of $p$-th powers of distances of all the points to this subspace. Sampling-based subset selection techniques require adaptive sampling iterations with multiple passes over the data. Matrix sketching techniques give a single-pass $(1+\eps)$ approximation for $\ell_{p}$ subspace approximation but require additional passes for subset selection.

In this work, we propose an MCMC algorithm to reduce the number of passes required by previous subset selection algorithms based on adaptive sampling. For $p=2$, our algorithm gives subset selection of nearly optimal size in only $2$ passes, whereas the number of passes required in previous work depend on $k$. Our algorithm picks a subset of size $\poly(k/\eps)$ that gives $(1+\eps)$ approximation to the optimal subspace.  The running time of the algorithm is $nd + d~ \poly(k/\eps)$. We extend our results to the case when outliers are present in the datasets, and suggest  a two pass algorithm for the same. Our ideas also extend to give a reduction in the number of passes required by adaptive sampling algorithms for $\ell_{p}$ subspace approximation and subset selection, for $p \geq 2$.

\end{abstract}

\input{intro}
\input{related_work}
\input{l_p_subspace}

\input{l_p_subspace_outliers}

\input{lower_bound}
\input{conclusion}

\bibliographystyle{plain}
\bibliography{reference} 
\end{document}

%% file: intro.tex
\section{Introduction}












Computing subspace approximation of large, high-dimensional input data is one of the most fundamental problems in data science and randomized numerical linear algebra. Given a dataset $\mathcal{X}=\{x_i\}_{i=1}^n: x_i \in \R^d$; where $n$ and $d$ are large, a positive integer $1<k\ll d$, and $1 \leq p < \infty$, the problem of $\ell_{p}$ subspace approximation is to find a $k$ dimensional linear subspace $V$ of $\R^{d}$ that minimizes the sum of $p$-th powers of the distances of all the points in the dataset to the subspace $V$, that is,
\[
 \err_{p}(\mathcal{X}, V):=\sum_{i=1}^{n} d(x_{i}, V)^{p}.
\]
This subspace approximation problem can also be seen as a dimensionality reduction as it compresses the data dimension. The subspace approximation problem for  $p=2$ (also known as low-rank matrix approximation) is well studied and can be solved exactly in time $O(\min \{n^2d, nd^2\})$ using the Singular value decomposition (SVD). However, computing exact SVD may not be practical when $n$ and $d$ are large. To overcome this, several faster algorithms have been proposed that closely approximate the optimal results obtained via SVD decomposition. We mention a few such notable results as follows. Frieze, Kannan, and Vempala~\cite{FKV} suggest computing low-rank approximation that gives an \textit{additive} approximation in time $O(nd\cdot \poly(k, 1/\eps))$ by sampling a subset of points with probability proportional to their squared lengths. This was later improved in~\cite{DBLP:conf/approx/DeshpandeV06,DeshpandeRVW06,DBLP:conf/focs/DeshpandeR10} to give a \textit{multiplicative} approximation guarantee by modifying and generalizing the squared-length sampling to sample multiple points adaptively over multiple passes. This line of work~\cite{FKV,DBLP:conf/approx/DeshpandeV06,DeshpandeRVW06,DBLP:conf/focs/DeshpandeR10} sample $\poly((k/\eps)^p)$ points, with the guarantee that their span contains a $k$-dimensional subspace that gives $(1+\eps)$ approximation to the optimum, with high probability. These sampling results are also known as \textit{column subset selection} for low-rank approximation.  Low-rank approximation based on row and column subset selection is more interpretable and advantageous as argued in~\cite{DrineasMM08,MahoneyD09,WangZ13,mahoney2011randomized,pmlr-v119-ida20a}. However, a limitation of most row and column subset selection algorithms that give multiple approximation guarantee is that they require multiple passes over the input matrix to perform multiple adaptive rounds of sampling, which can be impractical when $n$ and $d$ are large.
 
Another line of work follows deterministic sketching based techniques~\cite{Liberty13,GhashamiP14,GhashamiLPW16,CormodeDW18}, requires one passes over the data, and offers multiplicative approximation guarantee. The work due to~\cite{Sarlos06, ClarksonW13} gives randomized sketching for multiplicative approximation guarantee. However, a limitation of these results is that they don’t provide subset selection.
In this work, we focus on the problem of achieving a multiplicative approximation guarantee for $\ell_p$ norm via subset selection and simultaneously aim to minimize the number of passes over the input required by the adaptive sampling. We summarise our contributions as follows:

\subsection{Our results:} Our main contribution lies in minimizing the number of rounds required by the adaptive sampling to sample a subset of points whose span contains a $k$ dimensional subspace that offer multiplicative approximation for the $\ell_p$ error.  Our main ingredient is a carefully designed MCMC sampling distribution that ensures that the implied probability of sampling a point using MCMC distribution is sufficiently close to that of adaptive sampling distribution, and as a consequence, the $\ell_p$ error corresponding to these two sampling distributions is sufficiently close. We summarise our key contribution as follows:

\begin{itemize}
\item For $p=2$, whereas previous subset selection for nearly optimal approximation required $O(k\log k)$ passes by adaptive sampling \cite{DBLP:conf/approx/DeshpandeV06} or $O(\log k)$ passes by a combination of approximate volume sampling and adaptive sampling \cite{DBLP:conf/focs/DeshpandeR10}, our MCMC sampling algorithm can achieve the same in only $2$ passes over the data. We sample a set of $\poly(k/\eps)$ points whose span contains a $k$ dimensional subspace that offers $(1+\eps)$-multiplicative approximation to the optimal. The running time of our algorithm is $nd+d\cdot\poly(k/\eps)$. 
\item For $p=2$, our results generalize to subspace approximation with outliers, and offer a similar multiplicative approximation guarantee. This is an improvement over the recently proposed algorithms~\cite{DBLP:journals/tcs/DeshpandeP21,DeshpandeP20} by reducing their number of passes from $O(k \log k)$ to $2$.
\item Our ideas extend to $\ell_{p}$ subspace approximation, for $p \geq 2$, and give a significant reduction in the number of passes required by adaptive sampling algorithms for subspace approximation and subset selection.
\end{itemize}


%% file: related_work.tex
\section{Previous work}
We broadly split the baselines to compare with our proposal in the following  subsections. 
\subsection{Subset selection}




We first suggest baselines for $p=2$. 
 Frieze, Kannan, and Vempala~\cite{FKV} suggest a sampling algorithm that picks points proportional to their squared norm (called \textit{squared-length sampling}) and offers the guarantee that the span of $\poly(k/\eps)$ sampled points contains a $k$-dimensional subspace which offers additive approximation guarantee to the optimal $k$-dimensional subspace. However, their additive error is essentially $\eps$ times the Frobenius norm of input points, which can be arbitrarily large as compared to the true error. This result was strengthened by Deshpande and Vempala~\cite{DBLP:conf/approx/DeshpandeV06} by following squared-length sampling in an adaptive manner. They show that a subset of $\poly(k/\eps)$ points sampled using their distribution gives a multiplicative approximation guarantee. However, a limitation of their approach is that it requires taking $O(k\log k)$ passes over the input to generate the sampling distribution. Another way to achieve the multiplicative approximation guarantee was proposed by doing \textit{volume sampling}~\cite{DeshpandeRVW06}, and then following adaptive sampling in $O(k \log k)$ rounds~\cite{DBLP:conf/approx/DeshpandeV06}. The running time of their algorithm is $O(\nnz(\mathcal{X})\cdot k/\eps)$. The result due to Guruswami \textit{et. al.}
~\cite{DBLP:conf/soda/GuruswamiS12} show that sampling $O(k/\eps)$ points by volume sampling gives a bi-criteria $(1+\eps)$ approximation to the optimal subspace. 
 The result due to~\cite{DeshpandeV07} extends the above for $p\neq 2$, and suggests multiplicative subspace approximation. They first perform \textit{approximate volume sampling} that gives $k!(k+1)$ multiplicative guarantee. Then they follow  an adaptive sampling of $O(k\log k)$ rounds, and   sample a subset of $\tilde{O}(k^2 (k/\eps)^{p+1})$ with the guarantee that their span gives $(1+\eps)$ approximation, with high probability. Their running time is $\tilde{O}(nd\cdot k^3 (k/\eps)^{p+1})$. The first step of ~\cite{DeshpandeV07} mentioned above was recently improved to to $(k+1)$ by \cite{ChierichettiG0L17}. This was further strengthened  to $O((k+1)^{1/p})$ for $p\in(1, 2)$ and $O({(k+1)}^{1-1/p})$ for $p\geq 2$ due to~\cite{DanWZZR19}. Of course, here also, adaptive sampling of {$O(k \log k)$} rounds need to perform in order to achieve  $(1+\eps)$ multiplicative approximation. An advantage of these results is that they perform subset selection for subspace approximation, whereas the downside is that they require a large number of passes over the input to offer the multiplicative approximation guarantee.

\subsection{Frequent directions and sketching}




Liberty's deterministic matrix sketching \cite{Liberty13} suggests additive rank $k$ approximation guarantee by taking  $1$-pass over the data stream. A subsequent work due to Ghashami \textit{et.al.}\cite{GhashamiP14} provide a faster, deterministic algorithm that runs in $O\left(nd \cdot \text{poly}(k/\eps)\right)$ time and gives a multiplicative $(1+\eps)$-approximation to the optimum. A faster algorithm with running time $O(nd\cdot (k/\eps)^{2})$ time that offers the same guarantee was provided in~\cite{GhashamiLPW16}. An advantage of these results is that they are deterministic and require only one pass over the data. However, a major limitation of these results is that they work only for $p=2$, can not perform subset selection, and also do not extend when outliers are present in the datasets. Recently proposed work~\cite{CormodeDW18} extends the above mentioned deterministic sketching results for $p\neq 2$. Their Theorem $4.2$ suggests a one pass deterministic sketching for the $\ell_p$ (with $p\neq 2, \infty$) subspace approximation in streaming settings.  The running time of their algorithm is  $\poly(n^{\gamma}, d)$ with the update time $O(n^{\gamma}d)$, and offers $1/d^{\gamma}$ approximation factor, where $\gamma \in (0, 1)$ is a constant. Again limitation of this result is that they do not provide subset selection. Further Theorem $5.1$ of their result suggest an $\ell_1$ subspace approximation algorithm that offer $\poly(k)$ approximation factor, with update time $\poly(n, d)$ and space $n^{\gamma}\poly(d)$. The $\ell_1$ distance is known to be robust to the outliers and these results can potentially be used to give $\ell_p$ subspace approximation with outliers. However, they don’t provide any concrete theoretical guarantee on this statement. Whereas our results extend to $\ell_p$ subspace approximation with outliers under the assumption that error in inliers over $\ell_p$ norm is at least a constant fraction of $\ell_p$ error over all the points. 

There is another line of work based on the randomized sketching based techniques. Exploiting the random projection-based techniques~\cite{Sarlos06} suggests $(1+\eps)$ multiplicative approximation in running time $O(\nnz(X) \cdot \text{poly}(k/\eps))$. This was later improved by Clarkson and Woodruff~\cite{ClarksonW13} that offers the same multiplicative guarantee in a  faster running time  $O(\text{nnz}(X) + (n+d)\cdot \text{poly}(k/\eps))$. Again the advantage of these results is that they require only one pass over the input but they can not perform subset selection, and also not known to be robust. 

\subsection{Robust subspace approximation}
Our subspace approximation with outliers results can be seen as an improvement over the work due to Deshpande and Pratap~\cite{DBLP:journals/tcs/DeshpandeP21,DeshpandeP20} in the sense that it requires the lesser number passes over the dataset to offer multiplicative approximation guarantee over the inliers. For $p=2$, they require $O(k\log k)$ passes over the datasets whereas we require only two number of  passes over the input. Both these results require the assumption on the datasets that $\ell_p$ error over inliers is at least a constant fraction of $\ell_p$ error over all the points (see Equation~\eqref{eq:assumption}). Bhaskara and Kumar~\cite{BK2018} suggest a bicriteria approximation algorithm for subspace space approximation with outliers for $p=2$. However, their results require a stronger assumption on the datasets called \textit{rank-$k$ condition}, discard a large number of outliers than the optimal solution, and require an initial guess on the optimal error over inliers for their algorithm. Moreover, they don’t provide subset selection. 
Hardt and Moitra \cite{HardtM13} considered a related problem called \textit{robust subspace recovery} and gave an efficient algorithm for the problem  under a strong assumption on the data that requires any $d$ or fewer outliers to be linearly independent. 

The MCMC sampling has also been explored in the context of $k$-means clustering. The $D^2$-sampling proposed by Arthur and Vassilvitskii~\cite{k-meanspp} adaptively samples $k$ points -- one point in each passes over the input, and the sampled points give $O(\log k)$ approximation with respect to the optimal clustering solution. The results due to~\cite{DBLP:conf/nips/BachemLH016, BachemLHK16} suggest generating MCMC sampling distribution by taking only one pass over the input that closely approximates the desired $D^2$ sampling distribution, and offer close to the optimal clustering solution.

%% file: l_p_subspace.tex
\section{Background}\label{sec:background}
\noindent\textbf{$\ell_{p}$ subspace approximation with outliers:} Given a set of points $\mathcal{X}=\{x_{i}\}_{i=1}^n \in \R^{d}$, an integer $1 \leq k \leq d$, $1 \leq p < \infty$, and an upper bound on the fraction of outliers $0 \leq \beta \leq 1$, the problem is to find a $k$-dimensional linear subspace $V$ that minimizes the sum of $p$-th powers of distances of the $(1-\beta)n$ points nearest to it. If $N_{\beta}(V) \subseteq [n]$ denotes the set  of the indices of the nearest $(1-\beta)n$ points to $V$ among $x_{1}, x_{2}, \dotsc, x_{n}$, then  we want to minimize the following:
$
\sum_{i \in N_{\beta}(V)} d(x_{i}, V)^{p}.
$

\section{Subspace approximation in fewer passes by MCMC sampling}


We present the pseudocode of our algorithm in Algorithm \ref{alg:MCMC_sampling_lp}. Our algorithm starts with an initial subset $S_{0}$ and then uses it as a \emph{pivot} subset to approximate subsequent iterations of adaptive sampling in a single pass by an MCMC sampling procedure. In other words, $l$ iterations of adaptive sampling, with $t$ i.i.d. points to be picked in each iteration, requires $l$ passes over the data $\mathcal{X}$ to update the set w.r.t. which we need to do adaptive sampling. Our MCMC sampling algorithm uses a pivot subset and a random walk to approximate the distribution over $l$ iterations using only a single pass, for any $l$. Using the right choice of parameters $t, l$ and the length of the random walk $m$, we can do subset selection of near-optimal size for $\ell_{p}$ subspace approximation without requiring a large number of passes over the data. The number of passes required in previous randomized algorithm for subset selection depend on $l$, and we remove this dependence by MCMC sampling.

\begin{algorithm}[H]
\DontPrintSemicolon
  \KwInput{Data set of $n$ points $\mathcal{X} \subseteq \R^d$; a positive integer $k \leq d$; real-valued parameters $\eps, \alpha > 0$; positive integer parameters $t, l, m$}
  \KwOutput{$S_{0} \cup A_{1} \cup A_{2} \cup \dotsc \cup A_{l}$}
  {
  Sample a subset $S_{0}$ of $k$ point from $\mathcal{X}$ using $\alpha$-approximate volume sampling, i.e., the probability of picking any subset $S$ is
  \[
  \prob{S} \leq \frac{\alpha \cdot \mathrm{vol}(\Delta_{S})^p}{\sum_{T \suchthat \size{T}=k}\mathrm{vol}(\Delta_{T})^{p}},
  \]
  where $\Delta({S})$ denotes the simplex formed by the points in $S$ and the origin.
  }\label{line:volume_sampling_lp} \\
  \For {$i=1, 2, \dotsc, l$} 
  {
  Sample $x \in \mathcal{X}$ with probability $q(x)$, where
$q(x) = \dfrac{1}{2}\dfrac{d(x, \linspan{S_{0}})^{p}}{\err_{p}(\mathcal{X}, S_{0})} + \dfrac{1}{2 \size{\mathcal{X}}}$. \\
  Let $p(x) = \dfrac{d(x, \linspan{S_{0} \cup A_{1} \cup \dotsc \cup A_{i-1}})^{p}}{\err_{p}(\mathcal{X}, S_{0} \cup A_{1} \cup \dotsc \cup A_{i-1})}$ \\
   \While{$|A_i| \leq t$}
   {
    \For {$j=2, 3, \ldots m $}
    {
    Sample $y \in \mathcal{X}$ with probability $q(y)$. \\
    $p(y) \leftarrow \dfrac{d(y, \linspan{S_{0} \cup A_{1} \cup \dotsc \cup A_{i-1}})^{p}}{\err_{p}(\mathcal{X}, S_{0} \cup A_{1} \cup \dotsc \cup A_{i-1})}$ \\
    $\mathrm{Unif}(0, 1)$ be a uniform random number in $(0, 1)$ interval. \\
        \If{$\dfrac{p(y) q(x)}{p(x) q(y)}> \mathrm{Unif}(0, 1)$}
        {
            $x \leftarrow y$\\
            $p(x) \leftarrow p(y)$\\
        }
        
     }
     $A_i \longleftarrow A_i \cup \{x\}$\\
  }
  }\label{alg:line2p}
\caption{Subset selection for $\ell_{p}$ subspace approximation by MCMC sampling}\label{alg:MCMC_sampling_lp}
\end{algorithm}

\newcommand{\Sam}{S}
\newcommand{\mcmc}{\mathcal{A}^{S_0}(\Sam, l)}
\newcommand{\mcmcone}{\mathcal{A}^{S_0}(\Sam\cup\{x\}, l-1)}
\newcommand{\mcmczero}{\mathcal{A}^{S_0}(\Sam, 0)}

\newcommand{\comp}{\mathcal{Q}^{S_0}(\Sam, l)}
\newcommand{\compone}{\mathcal{Q}^{S_0}(\Sam\cup\{x\}, l-1)}
\newcommand{\compzero}{\mathcal{Q}^{S_0}(\Sam, 0)}

\newcommand{\sqlen}{\mathcal{B}(\Sam, l)}
\newcommand{\sqlenone}{\mathcal{B}(\Sam\cup\{x\}, l-1)}
\newcommand{\sqlenzero}{\mathcal{B}(\Sam, 0)}


\newcommand{\mcmcT}{\mathcal{A}^{S_0}(\Sam\cup\{T\}, l-1)}
\newcommand{\compT}{\mathcal{Q}^{S_0}(\Sam\cup\{T\}, l-1)}
\newcommand{\sqlenT}{\mathcal{B}(\Sam\cup\{T\}, l-1)}

\subsection{Expected error over multiple iterations of MCMC sampling}
%
%
%
First, let's set up the notation required to analyze adaptive sampling as well as the MCMC sampling in Algorithm \ref{alg:MCMC_sampling_lp}. For any fixed subset $S \subseteq \mathcal{X}$, we define
\begin{align*}
\err_{p}(\mathcal{X}, S) & = \sum_{x \in \mathcal{X}} d(x, S)^{p}, \\
P^{(1)}(T|S) & = \prod_{x \in T} \frac{d(x, \linspan{S})^{p}}{\err_{p}(\mathcal{X}, S)}, \qquad \text{for any subset $T$ of size $t$}, \\
\expect{T}{\err_{p}(\mathcal{X}, S \cup T)} & = \sum_{T \suchthat \size{T} = t} P^{(1)}(T \given S)~ \err_{p}(\mathcal{X}, S \cup T).
\end{align*}
Given a subset $S \subseteq \mathcal{X}$, $P^{(1)}(T \given S)$ denotes the probability of picking a subset $T \subseteq \mathcal{X}$ of i.i.d. $t$ points by adaptive sampling w.r.t. $S$. We use $P^{(l)}(T_{1:l}|S)$ to denote the probability of picking a subset $T_{1:l} = B_{1} \cup B_{2} \cup \dotsc \cup B_{l} \subseteq \mathcal{X}$ of $tl$ points by $l$ iterative rounds of adaptive sampling, where in the first round we sample a subset $B_{1}$ consisting of i.i.d. $t$ points w.r.t. $S$, in the second round we sample a subset $B_{2}$ consisting of i.i.d. $t$ points w.r.t. $S \cup B_{1}$, and so on to pick $T_{1:l} = B_{1} \cup B_{2} \cup \dotsc \cup B_{l}$ over $l$ iterations. Similarly, in the context of adaptive sampling, we use $T_{2:l}$ to denote $B_{2} \cup \dotsc \cup B_{l}$. We abuse the notation $\expect{T_{1:l} \given S}{\cdot}$ to denote the expectation over $T_{1:l}$ picked in $l$ iterative rounds of adaptive sampling starting from $S$.

Given a \emph{pivot} subset $\tilde{S} \subseteq \mathcal{X}$ and another subset $S \subseteq \mathcal{X}$ such that $\tilde{S} \subseteq S$, consider the following MCMC sampling with parameters $l, t, m$ that picks $l$ subsets $A_{1}, A_{2}, \dotsc, A_{l}$ of $t$ points each. $m$ denotes the number of steps of a random walk used to pick these points. This sampling can be implemented in a single pass over $\mathcal{X}$, for any $l, t, m$ and any given subsets $\tilde{S} \subseteq S$. For $T_{1:l} = A_{1} \cup A_{2} \cup \dotsc \cup A_{l}$ We use $\tilde{P}^{(l)}_{m}(T_{1:l} \given S)$ to denote the probability of picking $T_{1:l}$ as the output of the following sampling procedure. Similarly, in the context of MCMC sampling, we use $T_{2:l}$ to denote $A_{2} \cup \dotsc \cup A_{l}$. We abuse the notation $\expectilde{T_{1:l} \given S}{\cdot}$ to denote the expectation over $T_{1:l}$ picked using the MCMC sampling procedure starting from $S$ with a pivot subset $\tilde{S} \subseteq S$.

\begin{tcolorbox}[width=\textwidth,colback=black!5!white,colframe=black!75!black]
\noindent For $i=1, 2, \dotsc, l$ do:
\begin{enumerate}[(1)]
\item Sample $x \in \mathcal{X}$ with probability $q(x)$, where $q(x) = \dfrac{1}{2}\dfrac{ d(x, \linspan{\tilde{S}})^{p}}{\err_{p}(\mathcal{X}, \tilde{S})} + \dfrac{1}{2 \size{\mathcal{X}}}$.
\item Let $p(x) = \dfrac{d(x, \linspan{S \cup A_{1} \cup \dotsc \cup A_{i-1}})^{p}}{\err_{p}(\mathcal{X}, S \cup A_{1} \cup \dotsc A_{i-1})}$.
\item $A_{i} \leftarrow \emptyset$. While $\size{A_i} \leq t$ do:
\begin{enumerate}[(a)]
\item For $j = 2, 3, \dotsc, m$ do:
\begin{enumerate}[(i)]
\item Sample $y \in \mathcal{X}$ with probability $q(y)$, \\ and let $p(y) = \dfrac{d(y, \linspan{S \cup A_{1} \cup \dotsc \cup A_{i-1}})^{p}}{\err_{p}(\mathcal{X}, S \cup A_{1} \cup \dotsc \cup A_{i-1})}$.
\item If $\dfrac{p(y) q(x)}{p(x) q(y)}> \mathrm{Unif}(0, 1)$ then $x \leftarrow y$ and $p(x) \leftarrow p(y)$
\end{enumerate}
\item $A_i \leftarrow A_i \cup \{x\}$\\
\end{enumerate}
\end{enumerate}
\textbf{Output:} $T_{1:l} = A_{1} \cup A_{2} \cup \dotsc \cup A_{l}$.
\end{tcolorbox}
We require the following additional notation in our analysis of the above MCMC sampling. We use $\tilde{P}^{(1)}_{m}(T \given S)$ to denote the resulting distribution over subsets $T$ of size $t$, when we use the above sampling procedure with $l=1$. We define
\begin{align*}
\ind_{p}(\mathcal{X}, S) & = \mathbbm{1}\left(\err_{p}(\mathcal{X}, S) \leq \epsilon_{1}~ \err_{p}(\mathcal{X}, \tilde{S})\right), \\    
\expectilde{T}{\err_{p}(\mathcal{X}, S \cup T)} & = \sum_{T \suchthat \size{T} = t} \tilde{P}^{(1)}_{m}(T \given S)~ \err_{p}(\mathcal{X}, S \cup T), \\    
\expectilde{T}{\ind_{p}(\mathcal{X}, S \cup T)} & = \sum_{T \suchthat \size{T} = t} \tilde{P}^{(1)}_{m}(T \given S)~ \ind_{p}(\mathcal{X}, S \cup T).    
\end{align*}
Lemma \ref{lem:tv-dist} below shows that for any subset $S \subseteq \mathcal{X}$ that contains the pivot subset $\tilde{S}$ used for MCMC sampling, either $\err_{p}(\mathcal{X}, S)$ is small compared to $\err_{p}(\mathcal{X}, \tilde{S})$ or the adaptive sampling distribution w.r.t. $S$ is closely approximated in total variation distance by the MCMC sampling procedure with the pivot subset $\tilde{S}$.
\begin{lemma} \label{lem:tv-dist}
Let $\eps_1, \eps_2 \in (0, 1)$ and $\tilde{S} \subseteq S \subseteq \mathcal{X}$. Then for $m \geq 1+ \frac{2}{\eps_1}\log \tfrac{1}{\eps_2}$, either $\err_{p}(\mathcal{X}, S) \leq \eps_{1}~ \err_{p}(\mathcal{X}, \tilde{S})$ or $\norm{P^{(1)}(\cdot \given S) - \tilde{P}^{(1)}_{m}(\cdot \given S)}_{TV} \leq \eps_{2}t$.
\end{lemma}
\begin{proof}
Consider the $l=1, t=1$ case of the above MCMC sampling procedure. In this case, the procedure outputs only one element of $\mathcal{X}$. This random element is picked by $m$ steps of the following random walk. We first pick $x$ with probability $q(x)$, then we sample a $y$ with probability $q(y)$, we compute $p(x), p(y)$ and sample a number uniformly at random from the interval $(0, 1)$, and if $p(y) q(x)/p(x) q(y) > \mathrm{Unif}(0,1)$, then the random walk moves from $x$ to $y$ and updates $p(x)$ as $p(y)$. Observe that the stationary distribution of the above random walk is the adaptive sampling distribution w.r.t. $S$ given by $p(x) = d(x, \linspan{S})^{p}/\err_{p}(\mathcal{X}, S)$. Using Corollary $1$ of \cite{cai}, the total variation distance after $m$ steps of the random walk is bounded by  
\[
\left(1 - \frac{1}{\gamma}\right)^{m-1} \leq e^{-(m-1)/\gamma} \leq \eps_{2}, \quad \text{where $\gamma = \max_{x \in \mathcal{X}} \frac{p(x)}{q(x)}$}.
\]
The above bound is at most $\eps_{2}$ if we choose to run the random walk for $m \geq 1+\gamma \log \frac{1}{\eps_2}$ steps. Now suppose $\err_{p}(\mathcal{X}, S) > \eps_{1}~ \err_{p}(\mathcal{X}, \tilde{S})$. Then, for any $x \in \mathcal{X}$
\[
\frac{p(x)}{q(x)} = \dfrac{\dfrac{d(x, \linspan{S})^{p}}{\err_{p}(\mathcal{X}, S)}}{\dfrac{1}{2}\dfrac{d(x, \linspan{\tilde{S}})^{p}}{\err_{p}(\mathcal{X}, \tilde{S})} + \dfrac{1}{2 \size{\mathcal{X}}}} \leq \dfrac{2~ d(x, \linspan{S})^{p}~ \err_{p}(\mathcal{X}, \tilde{S})}{d(x, \linspan{\tilde{S}})^{p}~ \err_{p}(\mathcal{X}, S)} \leq \frac{2}{\eps_{1}},
\]
using $\tilde{S} \subseteq S$ and the above assumption $\err_{p}(\mathcal{X}, S) > \eps_{1}~ \err_{p}(\mathcal{X}, \tilde{S})$. Therefore, $m > \frac{2}{\eps_{1}} \log\frac{1}{\eps_{2}}$ ensures that $m$ steps of the random walk gives a distribution within total variation distance $\eps_{2}$ from the adaptive sampling distribution for picking a single point.

Note that for $t > 1$ both the adaptive sampling and the MCMC sampling procedure pick an i.i.d. sample of $t$ points, so the total variation distance is additive in $t$, which means
\[
\norm{P^{(1)}(\cdot \given S) - \tilde{P}^{(1)}_{m}(\cdot \given S)}_{TV} \leq \eps_{2}t,
\]
assuming $\err_{p}(\mathcal{X}, S) > \eps_{1}~ \err_{p}(\mathcal{X}, \tilde{S})$. This completes the proof the lemma.
\end{proof}

%

Now Lemma \ref{lemma:induction} analyzes the effect of starting with an initial subset $S_{0}$ and using the same $S_{0}$ as a pivot subset for doing the MCMC sampling for $l$ subsequent iterations of adaptive sampling, where we pick $t$ i.i.d. points in each iteration using $t$ independent random walks of $m$ steps. Lemma \ref{lemma:induction} shows that the expected error for subspace approximation after doing the $l$ iterations of adaptive sampling is not too far from the expected error for subspace approximation after replacing the $l$ iterations with MCMC sampling.
\begin{lemma} \label{lemma:induction}
For any subset $S_{0} \subseteq \mathcal{X}$, any $\eps_{1}, \eps_{2} \in (0, 1)$ and any positive integers $t, l, m$ with $m \geq 1+\frac{2}{\eps_2}\log \frac{1}{\eps_1}$,
\[
\expectilde{T_{1:l} \given S_{0}}{\err_{p}(\mathcal{X}, S_{0} \cup T_{1:l})} \leq \expect{T_{1:l} \given S_{0}}{\err_{p}(\mathcal{X}, S_{0} \cup T_{1:l})} + \left(\epsilon_{1} + \epsilon_{2} t l\right) \err_{p}(\mathcal{X}, S_{0}).
\]
\end{lemma}
\begin{proof}
We show a slightly stronger inequality than the one given above, i.e., for any $S_{0}$ such that $\tilde{S} \subseteq S_{0}$,
\begin{align*}
& \expectilde{T_{1:l} \given S_{0}}{\err_{p}(\mathcal{X}, S_{0} \cup T_{1:l})} \\
& \leq \expect{T_{1:l} \given S_{0}}{\err_{p}(\mathcal{X}, S_{0} \cup T_{1:l})} +  \left(\epsilon_{1} \expectilde{T_{1:l} \given S_{0}}{\ind_{p}(\mathcal{X}, S_{0} \cup T_{1:l})} + \epsilon_{2} t l\right) \err_{p}(\mathcal{X}, \tilde{S}).
\end{align*}
The special case $S_{0} = \tilde{S}$ gives the lemma. We prove the above-mentioned stronger statement by induction on $l$. For $l=0$, the above inequality holds trivially. Now assuming induction hypothesis, the above holds true for $l-1$ iterations (instead of $l$) starting with any subset $S_{1} = S_{0} \cup A \subseteq \mathcal{X}$ because $\tilde{S} \subseteq S_{0} \subseteq S_{1}$. 
\begin{align}
& \expectilde{T_{1:l} \given S_{0}}{\err_{p}(\mathcal{X}, S_{0} \cup T_{1:l})} \nonumber \\
& = \expectilde{S_{1} \given S_{0}}{\expectilde{T_{2:l} \given S_{1}}{\err_{p}(\mathcal{X}, S_{1} \cup T_{2:l})}} \nonumber \\
& = \sum_{S_{1} \suchthat \ind_{p}(\mathcal{X}, S_{1}) = 1} \tilde{P}^{(1)}_{m}(S_{1} \given S_{0})~ \expectilde{T_{2:l} \given S_{1}}{\err_{p}(\mathcal{X}, S_{1} \cup T_{2:l})} \nonumber \\
& \qquad \qquad + \sum_{S_{1} \suchthat \ind_{p}(\mathcal{X}, S_{1}) = 0} \tilde{P}^{(1)}_{m}(S_{1} \given S_{0})~ \expectilde{T_{2:l} \given S_{1}}{\err_{p}(\mathcal{X}, S_{1} \cup T_{2:l})}. \label{eq:parts-by-ind}
\end{align}
If $\ind_{p}(\mathcal{X}, S_{1}) = 1$ then $\err_{p}(\mathcal{X}, S_{1} \cup T_{2:l}) \leq \err_{p}(\mathcal{X}, S_{1}) \leq \eps_{1}~ \err_{p}(\mathcal{X}, S_{0})$, so the first part of the above sum can be bounded as follows.
\begin{align}
& \sum_{S_{1} \suchthat \ind_{p}(\mathcal{X}, S_{1}) = 1} \tilde{P}^{(1)}_{m}(S_{1} \given S_{0})~ \expectilde{T_{2:l} \given S_{1}}{\err_{p}(\mathcal{X}, S_{1} \cup T_{2:l})} \nonumber \\
& \leq \eps_{1}~ \err_{p}(\mathcal{X}, S_{0}) \sum_{S_{1} \suchthat \ind_{p}(\mathcal{X}, S_{1}) = 1} \tilde{P}^{(1)}_{m}(S_{1} \given S_{0})~ \expectilde{T_{2:l} \given S_{1}}{\ind_{p}(\mathcal{X}, S_{1} \cup T_{2:l})}. \label{eq:part-1-ind}
\end{align}
Now, the second part can be bounded as follows.
\begin{align}
& \sum_{S_{1} \suchthat \ind_{p}(\mathcal{X}, S_{1}) = 0} \tilde{P}^{(1)}_{m}(S_{1} \given S_{0})~ \expectilde{T_{2:l} \given S_{1}}{\err_{p}(\mathcal{X}, S_{1} \cup T_{2:l})} \nonumber \\
& = \sum_{S_{1} \suchthat \ind_{p}(\mathcal{X}, S_{1}) = 0} \tilde{P}^{(1)}_{m}(S_{1} \given S_{0})~ \expectilde{T_{2:l} \given S_{1}}{\err_{p}(\mathcal{X}, S_{1} \cup T_{2:l})} \nonumber \\
& \leq \sum_{S_{1} \suchthat \ind_{p}(\mathcal{X}, S_{1}) = 0} \tilde{P}^{(1)}_{m}(S_{1} \given S_{0}) \cdot \nonumber \\
& \qquad \left(\expect{T_{2:l} \given S_{1}}{\err_{p}(\mathcal{X}, S_{1} \cup T_{2:l})} +  \left(\eps_{1}~ \expectilde{T_{2:l} \given S_{1}}{\ind_{p}(\mathcal{X}, S_{1} \cup T_{2:l})} + \eps_{2} t (l-1)\right) \err_{p}(\mathcal{X}, \tilde{S})\right) \nonumber \\
& \qquad \qquad \text{by applying the induction hypothesis to $(l-1)$ iterations starting from $S_{1}$} \nonumber \\
& \leq \sum_{S_{1} \suchthat \ind_{p}(\mathcal{X}, S_{1}) = 0} P^{(1)}(S_{1} \given S_{0})~ \expect{T_{2:l} \given S_{1}}{\err_{p}(\mathcal{X}, S_{1} \cup T_{2:l})} \nonumber \\
& \qquad + \eps_{1}~ \err_{p}(\mathcal{X}, \tilde{S})~ \sum_{S_{1} \suchthat \ind_{p}(\mathcal{X}, S_{1}) = 0} \tilde{P}^{(1)}_{m}(S_{1} \given S_{0})~ \expectilde{T_{2:l} \given S_{1}}{\ind_{p}(\mathcal{X}, S_{1} \cup T_{2:l})} \nonumber \\
& \qquad + \eps_{2} t (l-1)~ \err_{p}(\mathcal{X}, \tilde{S})~ \sum_{S_{1} \suchthat \ind_{p}(\mathcal{X}, S_{1}) = 0} \tilde{P}^{(1)}_{m}(S_{1} \given S_{0}) \nonumber \\
& \qquad + \sum_{S_{1} \suchthat \ind_{p}(\mathcal{X}, S_{1}) = 0} \abs{\tilde{P}^{(1)}_{m}(S_{1} \given S_{0}) - P^{(1)}(S_{1} \given S_{0})}~ \expect{T_{2:l} \given S_{1}}{\err_{p}(\mathcal{X}, S_{1} \cup T_{2:l})} \nonumber \\
& \leq \sum_{S_{1}} P^{(1)}(S_{1} \given S_{0})~ \expect{T_{2:l} \given S_{1}}{\err_{p}(\mathcal{X}, S_{1} \cup T_{2:l})} \nonumber \\
& \qquad + \eps_{1}~ \err_{p}(\mathcal{X}, \tilde{S})~ \sum_{S_{1} \suchthat \ind_{p}(\mathcal{X}, S_{1}) = 0} \tilde{P}^{(1)}_{m}(S_{1} \given S_{0})~ \expectilde{T_{2:l} \given S_{1}}{\ind_{p}(\mathcal{X}, S_{1} \cup T_{2:l})} \nonumber \\
& \qquad + \eps_{2} t (l-1)~ \err_{p}(\mathcal{X}, \tilde{S}) + \sum_{S_{1} \suchthat \ind_{p}(\mathcal{X}, S_{1}) = 0} \abs{\tilde{P}^{(1)}_{m}(S_{1} \given S_{0}) - P^{(1)}(S_{1} \given S_{0})}~ \err_{p}(\mathcal{X}, \tilde{S}) \nonumber \\
& \leq \expect{T_{1:l} \given S_{0}}{\err_{p}(\mathcal{X}, S_{0} \cup T_{1:l})} \nonumber \\
& \qquad + \eps_{1}~ \err_{p}(\mathcal{X}, \tilde{S})~ \sum_{S_{1} \suchthat \ind_{p}(\mathcal{X}, S_{1}) = 0} \tilde{P}^{(1)}_{m}(S_{1} \given S_{0})~ \expectilde{T_{2:l} \given S_{1}}{\ind_{p}(\mathcal{X}, S_{1} \cup T_{2:l})} \nonumber \\
& \qquad + \eps_{2} t (l-1)~ \err_{p}(\mathcal{X}, \tilde{S})  + \norm{\tilde{P}^{(1)}(\cdot \given S_{0}) - P^{(1)}(\cdot \given S_{0})}_{TV}~ \err_{p}(\mathcal{X}, \tilde{S}) \\
& \leq \expect{T_{1:l} \given S_{0}}{\err_{p}(\mathcal{X}, S_{0} \cup T_{1:l})} \nonumber \\
& \qquad + \eps_{1}~ \err_{p}(\mathcal{X}, \tilde{S})~ \sum_{S_{1} \suchthat \ind_{p}(\mathcal{X}, S_{1}) = 0} \tilde{P}^{(1)}_{m}(S_{1} \given S_{0})~ \expectilde{T_{2:l} \given S_{1}}{\ind_{p}(\mathcal{X}, S_{1} \cup T_{2:l})} \nonumber \\
& \qquad + \eps_{2} t (l-1)~ \err_{p}(\mathcal{X}, \tilde{S})  + \eps_{2} t~ \err_{p}(\mathcal{X}, \tilde{S}), \label{eq:part-2-ind} 
\end{align}
using Lemma \ref{lem:tv-dist} about the total variation distance between $P^{(1)}$ and $\tilde{P}^{(1)}$ distributions. Plugging the bounds $\eqref{eq:part-1-ind}$ and $\eqref{eq:part-2-ind}$ into $\eqref{eq:parts-by-ind}$, we get
\begin{align*}
& \expectilde{T_{1:l} \given S_{0}}{\err_{p}(\mathcal{X}, S_{0} \cup T_{1:l})} \\
& \leq \expect{T_{1:l} \given S_{0}}{\err_{p}(\mathcal{X}, S_{0} \cup T_{1:l})} + \eps_{1}~ \err_{p}(\mathcal{X}, \tilde{S})~ \sum_{S_{1}} \tilde{P}^{(1)}_{m}(S_{1} \given S_{0})~ \expectilde{T_{2:l} \given S_{1}}{\ind_{p}(\mathcal{X}, S_{1} \cup T_{2:l})} \nonumber \\
& \qquad + \eps_{2} t (l-1)~ \err_{p}(\mathcal{X}, \tilde{S})  + \eps_{2} t~ \err_{p}(\mathcal{X}, \tilde{S}) \\
& = \expect{T_{1:l} \given S_{0}}{\err_{p}(\mathcal{X}, S_{0} \cup T_{1:l})} +  \left(\epsilon_{1} \expectilde{T_{1:l} \given S_{0}}{\ind_{p}(\mathcal{X}, S_{0} \cup T_{1:l})} + \epsilon_{2} tl\right)~ \err_{p}(\mathcal{X}, \tilde{S}) \\
& \leq \expect{T_{1:l} \given S_{0}}{\err_{p}(\mathcal{X}, S_{0} \cup T_{1:l})} +  \left(\epsilon_{1} + \epsilon_{2} t l\right)~ \err_{p}(\mathcal{X}, \tilde{S}),
\end{align*}
which completes the proof of Lemma \ref{lemma:induction}.
\end{proof}

\subsection{Near-optimal subset selection for $\ell_{2}$ subspace approximation in $2$ passes by MCMC sampling}
\begin{proposition} \label{prop:adaptive-iterations}
Let $k$ be any positive integer and let $\eps \in (0, 1)$. Let $S_{0}$ be any subset $S_{0} \subseteq \mathcal{X}$. If $S_{l} = S_{0} \cup T_{1:l}$ be obtained by starting from $S_{0}$ and doing  adaptive sampling in $l$ iterations, where in each iteration we add $t$ points from $\mathcal{X}$, then we have $\size{S_{l}} = \size{S_{0}} + tl$ and
\[
\expect{T_{1:l} \given S_{0}}{\err_{2}(\mathcal{X}, S_{0} \cup T_{1:l})} \leq \left(1 + \frac{2k}{t}\right)~ \err_{2}(\mathcal{X}, V^{*}) + \left(\frac{k}{t}\right)^{l}~ \err_{2}(\mathcal{X}, S_{0}),
\]
where $V^{*}$ minimizes $\err_{2}(\mathcal{X}, V)$ over all linear subspaces $V$ of dimension $k$.
\end{proposition}
\begin{proof}
Follows directly from Corollary 1 in~\cite{DBLP:conf/approx/DeshpandeV06}. 
\end{proof}

Now we state a slight modification of Theorem 1.3 from \cite{DeshpandeRVW06}. The original theorem gives $(k+1)$ approximation guarantee for exact volume sampling, in expectation. Proposition \ref{prop:l2-vol} below modifies it to say that $\alpha$-approximate volume sampling gives $\alpha (k+1)$ approximation guarantee, in expectation.
\begin{proposition}{(Immediate from Theorem 1.3, \cite{DeshpandeRVW06})} \label{prop:l2-vol}
If $S_{0} \subseteq \mathcal{X}$ is a random subset of size $k$ picked according to $\alpha$-approximate volume sampling probability, i.e.,
\[
\prob{S} \leq \frac{\alpha~ \operatorname{vol}(\Delta_{S})^{2}}{\sum_{T \suchthat \size{T}=k} \operatorname{vol}(\Delta_{T})^{2}},
\]
then $\err_{2}(\mathcal{X}, S_{0})$ is at most $\alpha (k+1)~ \err_{p}(\mathcal{X}, V^{*})$, in expectation, where $V^{*}$ minimizes $\err_{2}(\mathcal{X}, V)$ over all linear subspaces $V$ of dimension $k$..
\end{proposition}
Theorem \ref{thm:alpha-iterations} and Theorem \ref{thm:l2-algo} essentially show that the MCMC sampling in Algorithm \ref{alg:MCMC_sampling_lp} requires only $2$ passes to approximately implement the multiple passes of adaptive sampling. Moreover, using the guarantee in Lemma \ref{lemma:induction}, we get a nearly-optimal subset selection for the $p=2$ case of subspace approximation in only $2$ passes.
\begin{theorem} \label{thm:alpha-iterations}
Let $k$ be any positive integer and let $\eps \in (0, 1)$. Let $S_{0}$ be a subset of $k$ points sampled from $\mathcal{X}$ using $\alpha$-approximate volume sampling for $p=2$. Let $S_{l} = S_{0} \cup T_{1:l}$ be obtained by starting from $S_{0}$ and doing $l$ iterations of MCMC sampling, where in each iteration we add $t$ points from $\mathcal{X}$ by running $t$ independent random walks for $m$ steps each. Then for $t = 8k/\eps$, $l = \log(2\alpha(k+1)/\eps)/\log(8/\eps)$ and $m \geq 1 + 128k \alpha \log^{2}(2\alpha/\eps)/\eps^{2} \log(8/\eps)$, we have
\[
\expectilde{T_{1:l} \given S_{0}}{\err_{2}(\mathcal{X}, S_{0} \cup T_{1:l})} \leq (1+\eps)~ \err_{2}(\mathcal{X}, V^{*}), \]
where $V^{*}$ minimizes $\err_{2}(\mathcal{X}, V)$ over all linear subspaces $V$ of dimension $k$.
\end{theorem}
\begin{proof}
We use the following setting of parameters $t, l, \eps_{1}, \eps_{2}, m$ as functions of $k, \eps, \alpha$.
\begin{align*}
t & = \frac{8k}{\eps}, \qquad l = \frac{\log(2\alpha(k+1)/\eps)}{\log(8/\eps)}, \qquad
\eps_{1} = \frac{\eps}{8\alpha (k+1)}, \qquad \text{and} \\
\eps_{2} & = \frac{\eps}{8tl\alpha(k+1)} = \frac{\eps^{2} \log (8/\eps)}{64k(k+1) \alpha \log(2\alpha/\eps)}, \\
m & = 1 + \frac{2}{\eps_{2}} \log\frac{1}{\eps_{1}} = 1 + \frac{128k(k+1) \alpha \log^{2}(2\alpha/\eps)}{\eps^{2} \log(8/\eps)}
\end{align*}
We have
\begin{align*}
& \expectilde{T_{1:l} \given S_{0}}{\err_{p}(\mathcal{X}, S_{0} \cup T_{1:l})} \\
& \leq \expect{T_{1:l} \given S_{0}}{\err_{p}(\mathcal{X}, S_{0} \cup T_{1:l})} + \left(\epsilon_{1} + \epsilon_{2} t l\right)~ \err_{p}(\mathcal{X}, S_{0}) \\
& \hspace{6cm} \text{by Lemma \ref{lemma:induction} and $m \geq 1+\frac{2}{\eps_{2}}\log \frac{1}{\eps_{1}}$} \\
& \leq \left(1 + \frac{2k}{t}\right)~ \err_{2}(\mathcal{X}, V^{*}) + \left(\frac{k}{t}\right)^{l}~ \err_{2}(\mathcal{X}, S_{0}) +  \left(\epsilon_{1} + \epsilon_{2} t l\right)~ \err_{2}(\mathcal{X}, S_{0}) \\
& \hspace{6cm} \text{by Proposition \ref{prop:adaptive-iterations}} \\
& \leq \left(1 + \frac{\eps}{4}\right)~ \err_{2}(\mathcal{X}, V^{*}) + \frac{\eps}{2\alpha(k+1)} \alpha(k+1)~ \err_{2}(\mathcal{X}, V^{*}) \\
& \qquad \qquad + \left(\frac{\eps}{8\alpha(k+1)} + \frac{\eps}{8\alpha(k+1)}\right) \alpha (k+1)~ \err_{2}(\mathcal{X}, V^{*}) \\
& \hspace{6cm} \text{by Proposition \ref{prop:l2-vol} and plugging in $t, l, \eps_{1}, \eps_{2}$} \\
& \leq \left(1 + \frac{\eps}{4}\right)~ \err_{2}(\mathcal{X}, V^{*}) + \frac{\eps}{2}~ \err_{2}(\mathcal{X}, V^{*}) + \frac{\eps}{4}~ \err_{2}(\mathcal{X}, V^{*}) \\ 
& \leq (1+\eps)~ \err_{2}(\mathcal{X}, V^{*}).
\end{align*}
\end{proof}
It is known (Proposition 1 in \cite{DBLP:conf/approx/DeshpandeV06}) that picking a subset of size $k$ from $\mathcal{X}$ by adaptive sampling in $k$ iterations, with $1$ point per iteration, gives $\alpha$-approximate volume sampling with $\alpha = k! = 2^{O(k \log k)}$. Theorem \ref{thm:alpha-iterations} can be used to reduce the number of passes required in \cite{DBLP:conf/approx/DeshpandeV06} for a $(1+\eps)$-approximation from $O(k \log k)$ to $k+1$. Now we show that if we use a single-pass MCMC algorithm for $\alpha$-approximate volume sampling for a suitable $\alpha$, then the entire Algorithm \ref{alg:MCMC_sampling_lp} can be implemented in only two passes.  

\begin{theorem} \label{thm:l2-algo}
Let $k$ be any positive integer and let $\eps \in (0, 1)$. Algorithm \ref{alg:MCMC_sampling_lp} can be implemented in $2$ passes over $\mathcal{X}$, with $\alpha$-approximate volume sampling to pick a subset $S_{0}$ of size $k$ taking one pass over $\mathcal{X}$ and the remaining MCMC procedure with $S_{0}$ as its pivot subset taking another pass. Setting $t, l, m$ as $t = 8k/\eps$, $l = \log(2/\eps)/\log(8/\eps)$ and $m \geq 1 + 128k \log^{2}(2(k+2)/\eps)/\eps^{2} \log(8/\eps)$, the algorithm picks a subset of $tl = O(k \log k/\eps)$ points given by $T_{1:l} = A_{1} \cup A_{2} \cup \dotsc \cup A_{l}$ such that
\[
\expectilde{T_{1:l} \given S_{0}}{\err_{2}(\mathcal{X}, S_{0} \cup T_{1:l})} \leq (1+\eps)~ \err_{2}(\mathcal{X}, V^{*}),
\]
where $V^{*}$ minimizes $\err_{2}(\mathcal{X}, V)$ over all linear subspaces $V$ of dimension $k$. This gives a near-optimal $O(k \log k/\eps)$-sized subset selection with $(1+\eps)$ approximation guarantee in only $2$ passes over $\mathcal{X}$. The running time of MCMC procedure is $d\cdot \poly(k/\eps).$
\end{theorem}
\begin{proof}
Anari et al. \cite{anari2016monte} give an MCMC algorithm to sample approximately from volume sampling. They start with any subset $S$, pick $i \in S$ and $j \notin S$ uniformly at random, and perform a lazy random walk over $k$-sized subsets, i.e., move to $T = S \setminus \{i\} \cup \{j\}$ with probability $\frac{1}{2} \min\{1, \frac{\mathrm{vol}(\Delta_{T})^{2}}{\mathrm{vol}(\Delta_{S})^{2}}\}$, and with the remaining probability stay at $S$. They show that in $\text{poly}(n, k, \log(1/\delta \mathrm{vol}(\Delta_{S}))$ steps of the above random walk starting at $S$, the resulting distribution is within $\delta$ total variation distance from the exact volume sampling distribution. Thus, using Proposition \ref{prop:l2-vol} with $\alpha = 1$ and adding the error due to total variation distance, we get that  the expected error $\err_{2}(\mathcal{X}, S)$ at the end of $\text{poly}(n, k, \log(1/\delta \mathrm{vol}(\Delta_{S}))$ steps of the above random walk is at most $(k+1)~ \err_{p}(\mathcal{X}, V^{*}) + \delta \max_{S \suchthat \size{S}=k} \err_{p}(\mathcal{X}, S)$, where $V^{*}$ minimizes $\err_{2}(\mathcal{X}, V)$ over all linear subspaces $V$ of dimension $k$. We choose $\delta$ as follows.
\[
\delta = \frac{1}{\kappa} =  \frac{\sigma_{\min}}{\sigma_{\max}} \leq \frac{\min_{S \suchthat \size{S}=k} \err_{p}(\mathcal{X}, S)}{\max_{S \suchthat \size{S}=k} \err_{p}(\mathcal{X}, S)},
\]
where $\sigma_{\min}$ and $\sigma_{\max}$ are the minimum and the maximum singular values, respectively, of the $n$-by-$d$ matrix whose rows are all the points in $\mathcal{X}$, and $\kappa$ is the condition number. Then the expected error of $\err_{2}(\mathcal{X}, S)$ is at most $(k+2)~ \err_{p}(\mathcal{X}, V^{*})$. Observe that $\min_{S \suchthat \size{S}=k} \mathrm{vol}(\Delta_{S})^{2} \geq \sigma_{\min}^{k}$. Note that for the above choice of $\delta$ and any $k$-sized subset $S \subseteq \mathcal{X}$, we can upper bound $\log(1/\delta \mathrm{vol}(\Delta_{S})^{2})$ by $O(k \log \kappa)$, where $\kappa$ is the condition number of the $n$-by-$d$ matrix whose rows are the points of $\mathcal{X}$. Thus, in total number of $\text{poly}(n, k, \log \kappa)$ steps of the random walk we get an initial subset $S_{0}$ with $\err_{2}(\mathcal{X}, S_{0})$ upper bounded by $(k+2)~ \err_{2}(\mathcal{X}, S^{*})$, in expectation. 

Another alternative to implement approximate volume sampling is to do a $(1+\eps)$ volume-preseving random projection using only one pass over the given data \cite{DBLP:conf/approx/MagenZ08}, and then do exact volume sampling on the projected data in $\poly(k/\eps)$ dimensions. This gives another alternative to obtain the initial subset $S_{0}$. 

The rest of the proof using MCMC sampling algorithm starting from $S_{0}$ using the same subset as a pivot is similar to the proof of Theorem \ref{thm:alpha-iterations}.

The running time of MCMC procedure is $O(tlm \cdot d(tl)^2)=d\cdot\poly(k/\eps)$, where $d\cdot(tl)^2$ is the time required to project a point onto $\mathrm{span}(S)$,  where $S$ is a subset of $O(tl)$ points of dimension $d$ each.

\end{proof}

\subsection{$(k+1)$-pass subset selection for $\ell_{p}$ subspace approximation by MCMC sampling}
Our results can be extended in a similar manner for $p \geq 2$ by combining our MCMC sampling algorithm with the initialization $S_{0}$ given by adaptively sampling $k$ points in $k$ iterations, one point per iteration. Note that the initialization takes $k$ passes over the data.
\begin{theorem}[Theorem $1$ of~\cite{DeshpandeV07}]
Let $S_0$ be the subset of $k$ points sampled using volume sampling stated in line~\ref{line:volume_sampling_lp} of Algorithm~\ref{alg:MCMC_sampling_lp}. Then we have $\expect{S}{\err_{p}(\mathcal{X}, S)^{p}} \leq k! (k+1)^{p}~ \err_{p}(\mathcal{X}, S^{*})$, where $S^{*} = \underset{S \suchthat \size{S} = k}{\operatorname{argmin}}~ \err_{2}(\mathcal{X}, S)$.

\end{theorem}

%% file: l_p_subspace_outliers.tex
\section{$\ell_{2}$ subspace approximation with outliers in two passes by MCMC sampling}
We further extend our results for subspace approximation with outliers problem, which is essentially finding the optimal subspace only over the inliers.  We define it formally  in Section~\ref{sec:background}. Our result is an improvement of Theorem $12$ (for $p=2$) of~\cite{DBLP:journals/tcs/DeshpandeP21} in the sense that we reduce the number of passes required by the adaptive sampling. Their result requires one assumption on the input, and we state it as follows. Let $\mathcal{X}=\{x_i\}_{i=1}^n$ be a set of $n$ points in $d$-dimensional space, and $S^{*}_I$ is the optimal $k$-dimensional subspace over inlier set (denoted as $I$), then 
\begin{align}\label{eq:assumption}
\frac{\sum_{i \in I} d(x_{i}, S^{*}_I)^{2}}{\sum_{i=1}^{n} d(x_{i},S^{*}_I)^{2}}\geq \lambda, \qquad \text{where~} \lambda\in (0, 1), \text{~is a constant.}
\end{align}

We state our result as follows:

\begin{theorem} \label{thm:multi-line}
For any given set of points $\mathcal{X}=\{x_i\}_{i=1}^n \in \R^{d}$, let $I \subseteq [n]$ be the set of optimal $(1-\beta)n$ inliers and $S^{*}_I$ be the optimal subspace over inliers that minimizes their squared distance. Let the input satisfy the condition stated in Equation~\eqref{eq:assumption}, and $k$ be a positive integer, and $\eps\in(0, 1).$ Then Algorithm \ref{alg:MCMC_sampling_lp} can be implemented in $2$ passes over $\mathcal{X}$, with $\alpha$-approximate volume sampling to pick a subset $S_{0}$ of size $k$ taking one pass over $\mathcal{X}$ and the remaining MCMC procedure with $S_{0}$ as its pivot subset taking another pass. Setting $t, l, m$ as $t = O(k/\eps)$, $l = O(\log(\alpha k/\eps\lambda)/\log(1/\eps))$ and $m \geq 1 + O( (\alpha k/\lambda) \log^{2}(\alpha k/\eps\lambda)/\eps^{2} \log(1/\eps))$, the algorithm picks a subset of $tl = O(k^2 \log k/\eps^2)$ points given by $T_{1:l} = A_{1} \cup A_{2} \cup \dotsc \cup A_{l}$ such that
\[
\expectilde{T_{1:l} \given S_{0}}{\sum_{j\in N_{\beta}(S_{0} \cup T_{1:l})}d(x_j, S_{0} \cup T_{1:l})^2} \leq (1+\eps)~ \sum_{i\in I}d(x_i, S^{*}_I)^2, 
\]
  \text{where $S^{*}_I = \underset{S \suchthat \size{S} = k}{\operatorname{argmin}}~ \sum_{i\in I}d(x_i, S)^2$}, and $N_{\beta}(S_{0} \cup T_{1:l}) \subseteq [n]$ denotes the set  of the indices of the nearest $(1-\beta)n$ points to $(S_{0} \cup T_{1:l})$ among $x_{1}, x_{2}, \dotsc, x_{n}$.  This gives a near-optimal $O(k^2 \log k/\eps^2)$-sized subset selection with $(1+\eps)$ approximation guarantee over inliers in only $2$ passes over $\mathcal{X}$.
\end{theorem}
\begin{proof}
Follows  from Theorems $11, 12$ of~\cite{DBLP:journals/tcs/DeshpandeP21} and Theorem~\ref{thm:alpha-iterations} of this paper. 
\end{proof}

%% file: lower_bound.tex

%% file: conclusion.tex
\section{Conclusion}
We improve upon the previous work on sampling-based subspace approximation and subset selection algorithms that require adaptive sampling, and hence, multiple passes over the given data. This renders some of the algorithms based on adaptive sampling less practical on large data, especially if the number of passes required depends on the target dimension $k$. Our MCMC sampling algorithm can significantly reduce the number of passes required in the various applications of adaptive sampling for $\ell_{p}$ subspace approximation and subset selection.


%% file: lipics-v2021-sample-article.bbl
\begin{thebibliography}{10}

\bibitem{anari2016monte}
Nima Anari, Shayan~Oveis Gharan, and Alireza Rezaei.
\newblock Monte carlo markov chain algorithms for sampling strongly rayleigh
  distributions and determinantal point processes.
\newblock In {\em 29th Annual Conference on Learning Theory {(COLT)}},
  volume~49, pages 103--115. PMLR, 2016.
\newblock URL: \url{http://proceedings.mlr.press/v49/anari16.html}.

\bibitem{k-meanspp}
David Arthur and Sergei Vassilvitskii.
\newblock k-means++: the advantages of careful seeding.
\newblock In {\em Proceedings of the Eighteenth Annual {ACM-SIAM} Symposium on
  Discrete Algorithms, {SODA} 2007, New Orleans, Louisiana, USA, January 7-9,
  2007}, pages 1027--1035, 2007.
\newblock URL: \url{http://dl.acm.org/citation.cfm?id=1283383.1283494}.

\bibitem{BachemLHK16}
Olivier Bachem, Mario Lucic, S.~Hamed Hassani, and Andreas Krause.
\newblock Approximate k-means++ in sublinear time.
\newblock In {\em Proceedings of the Thirtieth {AAAI} Conference on Artificial
  Intelligence, February 12-17, 2016, Phoenix, Arizona, {USA}}, pages
  1459--1467, 2016.
\newblock URL:
  \url{http://www.aaai.org/ocs/index.php/AAAI/AAAI16/paper/view/12147}.

\bibitem{DBLP:conf/nips/BachemLH016}
Olivier Bachem, Mario Lucic, Seyed~Hamed Hassani, and Andreas Krause.
\newblock Fast and provably good seedings for k-means.
\newblock In {\em Advances in Neural Information Processing Systems 29: Annual
  Conference on Neural Information Processing Systems 2016, December 5-10,
  2016, Barcelona, Spain}, pages 55--63, 2016.
\newblock URL:
  \url{https://proceedings.neurips.cc/paper/2016/hash/d67d8ab4f4c10bf22aa353e27879133c-Abstract.html}.

\bibitem{BK2018}
Aditya Bhaskara and Srivatsan Kumar.
\newblock Low rank approximation in the presence of outliers.
\newblock In {\em Approximation, Randomization, and Combinatorial Optimization.
  Algorithms and Techniques, {APPROX/RANDOM} 2018, August 20-22, 2018 -
  Princeton, NJ, {USA}}, pages 4:1--4:16, 2018.
\newblock URL: \url{https://doi.org/10.4230/LIPIcs.APPROX-RANDOM.2018.4}, \href
  {http://dx.doi.org/10.4230/LIPIcs.APPROX-RANDOM.2018.4}
  {\path{doi:10.4230/LIPIcs.APPROX-RANDOM.2018.4}}.

\bibitem{cai}
Haiyan Cai.
\newblock Exact bound for the convergence of metropolis chains.
\newblock {\em Stochastic Analysis and Applications}, 18(1):63--71, 2000.
\newblock URL: \url{https://doi.org/10.1080/07362990008809654}, \href
  {http://arxiv.org/abs/https://doi.org/10.1080/07362990008809654}
  {\path{arXiv:https://doi.org/10.1080/07362990008809654}}, \href
  {http://dx.doi.org/10.1080/07362990008809654}
  {\path{doi:10.1080/07362990008809654}}.

\bibitem{ChierichettiG0L17}
Flavio Chierichetti, Sreenivas Gollapudi, Ravi Kumar, Silvio Lattanzi, Rina
  Panigrahy, and David~P. Woodruff.
\newblock Algorithms for {\textdollar}{\textbackslash}ell{\_}p{\textdollar}
  low-rank approximation.
\newblock In {\em Proceedings of the 34th International Conference on Machine
  Learning, {ICML} 2017, Sydney, NSW, Australia, 6-11 August 2017}, pages
  806--814, 2017.
\newblock URL: \url{http://proceedings.mlr.press/v70/chierichetti17a.html}.

\bibitem{ClarksonW13}
Kenneth~L. Clarkson and David~P. Woodruff.
\newblock Low rank approximation and regression in input sparsity time.
\newblock In {\em Symposium on Theory of Computing Conference, STOC'13, Palo
  Alto, CA, USA, June 1-4, 2013}, pages 81--90, 2013.
\newblock URL: \url{https://doi.org/10.1145/2488608.2488620}, \href
  {http://dx.doi.org/10.1145/2488608.2488620}
  {\path{doi:10.1145/2488608.2488620}}.

\bibitem{CormodeDW18}
Graham Cormode, Charlie Dickens, and David~P. Woodruff.
\newblock Leveraging well-conditioned bases: Streaming and distributed
  summaries in minkowski p-norms.
\newblock In {\em Proceedings of the 35th International Conference on Machine
  Learning, {ICML} 2018, Stockholmsm{\"{a}}ssan, Stockholm, Sweden, July 10-15,
  2018}, pages 1048--1056, 2018.
\newblock URL: \url{http://proceedings.mlr.press/v80/cormode18a.html}.

\bibitem{DanWZZR19}
Chen Dan, Hong Wang, Hongyang Zhang, Yuchen Zhou, and Pradeep Ravikumar.
\newblock Optimal analysis of subset-selection based l{\_}p low-rank
  approximation.
\newblock In {\em Advances in Neural Information Processing Systems 32: Annual
  Conference on Neural Information Processing Systems 2019, NeurIPS 2019,
  December 8-14, 2019, Vancouver, BC, Canada}, pages 2537--2548, 2019.
\newblock URL:
  \url{https://proceedings.neurips.cc/paper/2019/hash/80a8155eb153025ea1d513d0b2c4b675-Abstract.html}.

\bibitem{DeshpandeP20}
Amit Deshpande and Rameshwar Pratap.
\newblock Subspace approximation with outliers.
\newblock In {\em Computing and Combinatorics - 26th International Conference,
  {COCOON} 2020, Atlanta, GA, USA, August 29-31, 2020, Proceedings}, pages
  1--13, 2020.
\newblock URL: \url{https://doi.org/10.1007/978-3-030-58150-3\_1}, \href
  {http://dx.doi.org/10.1007/978-3-030-58150-3\_1}
  {\path{doi:10.1007/978-3-030-58150-3\_1}}.

\bibitem{DBLP:journals/tcs/DeshpandeP21}
Amit Deshpande and Rameshwar Pratap.
\newblock Sampling-based dimension reduction for subspace approximation with
  outliers.
\newblock {\em Theor. Comput. Sci.}, 858:100--113, 2021.
\newblock URL: \url{https://doi.org/10.1016/j.tcs.2021.01.021}, \href
  {http://dx.doi.org/10.1016/j.tcs.2021.01.021}
  {\path{doi:10.1016/j.tcs.2021.01.021}}.

\bibitem{DBLP:conf/focs/DeshpandeR10}
Amit Deshpande and Luis Rademacher.
\newblock Efficient volume sampling for row/column subset selection.
\newblock In {\em 51th Annual {IEEE} Symposium on Foundations of Computer
  Science, {FOCS} 2010, October 23-26, 2010, Las Vegas, Nevada, {USA}}, pages
  329--338, 2010.
\newblock URL: \url{https://doi.org/10.1109/FOCS.2010.38}, \href
  {http://dx.doi.org/10.1109/FOCS.2010.38} {\path{doi:10.1109/FOCS.2010.38}}.

\bibitem{DeshpandeRVW06}
Amit Deshpande, Luis Rademacher, Santosh~S. Vempala, and Grant Wang.
\newblock Matrix approximation and projective clustering via volume sampling.
\newblock {\em Theory Comput.}, 2(12):225--247, 2006.
\newblock URL: \url{https://doi.org/10.4086/toc.2006.v002a012}, \href
  {http://dx.doi.org/10.4086/toc.2006.v002a012}
  {\path{doi:10.4086/toc.2006.v002a012}}.

\bibitem{DeshpandeV07}
Amit Deshpande and Kasturi~R. Varadarajan.
\newblock Sampling-based dimension reduction for subspace approximation.
\newblock In {\em Proceedings of the 39th Annual {ACM} Symposium on Theory of
  Computing, San Diego, California, USA, June 11-13, 2007}, pages 641--650,
  2007.
\newblock URL: \url{https://doi.org/10.1145/1250790.1250884}, \href
  {http://dx.doi.org/10.1145/1250790.1250884}
  {\path{doi:10.1145/1250790.1250884}}.

\bibitem{DBLP:conf/approx/DeshpandeV06}
Amit Deshpande and Santosh~S. Vempala.
\newblock Adaptive sampling and fast low-rank matrix approximation.
\newblock In {\em Approximation, Randomization, and Combinatorial Optimization.
  Algorithms and Techniques, 9th International Workshop on Approximation
  Algorithms for Combinatorial Optimization Problems, {APPROX} 2006 and 10th
  International Workshop on Randomization and Computation, {RANDOM} 2006,
  Barcelona, Spain, August 28-30 2006, Proceedings}, pages 292--303, 2006.
\newblock URL: \url{https://doi.org/10.1007/11830924\_28}, \href
  {http://dx.doi.org/10.1007/11830924\_28} {\path{doi:10.1007/11830924\_28}}.

\bibitem{DrineasMM08}
Petros Drineas, Michael~W. Mahoney, and S.~Muthukrishnan.
\newblock Relative-error {CUR} matrix decompositions.
\newblock {\em {SIAM} J. Matrix Anal. Appl.}, 30(2):844--881, 2008.
\newblock URL: \url{https://doi.org/10.1137/07070471X}, \href
  {http://dx.doi.org/10.1137/07070471X} {\path{doi:10.1137/07070471X}}.

\bibitem{FKV}
Alan~M. Frieze, Ravi Kannan, and Santosh~S. Vempala.
\newblock Fast monte-carlo algorithms for finding low-rank approximations.
\newblock {\em J. {ACM}}, 51(6):1025--1041, 2004.
\newblock URL: \url{https://doi.org/10.1145/1039488.1039494}, \href
  {http://dx.doi.org/10.1145/1039488.1039494}
  {\path{doi:10.1145/1039488.1039494}}.

\bibitem{GhashamiLPW16}
Mina Ghashami, Edo Liberty, Jeff~M. Phillips, and David~P. Woodruff.
\newblock Frequent directions: Simple and deterministic matrix sketching.
\newblock {\em {SIAM} J. Comput.}, 45(5):1762--1792, 2016.
\newblock URL: \url{https://doi.org/10.1137/15M1009718}, \href
  {http://dx.doi.org/10.1137/15M1009718} {\path{doi:10.1137/15M1009718}}.

\bibitem{GhashamiP14}
Mina Ghashami and Jeff~M. Phillips.
\newblock Relative errors for deterministic low-rank matrix approximations.
\newblock In {\em Proceedings of the Twenty-Fifth Annual {ACM-SIAM} Symposium
  on Discrete Algorithms, {SODA} 2014, Portland, Oregon, USA, January 5-7,
  2014}, pages 707--717, 2014.
\newblock URL: \url{https://doi.org/10.1137/1.9781611973402.53}, \href
  {http://dx.doi.org/10.1137/1.9781611973402.53}
  {\path{doi:10.1137/1.9781611973402.53}}.

\bibitem{DBLP:conf/soda/GuruswamiS12}
Venkatesan Guruswami and Ali~Kemal Sinop.
\newblock Optimal column-based low-rank matrix reconstruction.
\newblock In {\em Proceedings of the Twenty-Third Annual {ACM-SIAM} Symposium
  on Discrete Algorithms, {SODA} 2012, Kyoto, Japan, January 17-19, 2012},
  pages 1207--1214, 2012.
\newblock URL: \url{https://doi.org/10.1137/1.9781611973099.95}, \href
  {http://dx.doi.org/10.1137/1.9781611973099.95}
  {\path{doi:10.1137/1.9781611973099.95}}.

\bibitem{HardtM13}
Moritz Hardt and Ankur Moitra.
\newblock Algorithms and hardness for robust subspace recovery.
\newblock In {\em {COLT} 2013 - The 26th Annual Conference on Learning Theory,
  June 12-14, 2013, Princeton University, NJ, {USA}}, pages 354--375, 2013.
\newblock URL: \url{http://proceedings.mlr.press/v30/Hardt13.html}.

\bibitem{pmlr-v119-ida20a}
Yasutoshi Ida, Sekitoshi Kanai, Yasuhiro Fujiwara, Tomoharu Iwata, Koh
  Takeuchi, and Hisashi Kashima.
\newblock Fast deterministic {CUR} matrix decomposition with accuracy
  assurance.
\newblock In Hal~Daumé III and Aarti Singh, editors, {\em Proceedings of the
  37th International Conference on Machine Learning}, volume 119 of {\em
  Proceedings of Machine Learning Research}, pages 4594--4603. PMLR, 13--18 Jul
  2020.
\newblock URL: \url{http://proceedings.mlr.press/v119/ida20a.html}.

\bibitem{Liberty13}
Edo Liberty.
\newblock Simple and deterministic matrix sketching.
\newblock In {\em The 19th {ACM} {SIGKDD} International Conference on Knowledge
  Discovery and Data Mining, {KDD} 2013, Chicago, IL, USA, August 11-14, 2013},
  pages 581--588, 2013.
\newblock URL: \url{https://doi.org/10.1145/2487575.2487623}, \href
  {http://dx.doi.org/10.1145/2487575.2487623}
  {\path{doi:10.1145/2487575.2487623}}.

\bibitem{DBLP:conf/approx/MagenZ08}
Avner Magen and Anastasios Zouzias.
\newblock Near optimal dimensionality reductions that preserve volumes.
\newblock In {\em Approximation, Randomization and Combinatorial Optimization.
  Algorithms and Techniques, 11th International Workshop, {APPROX} 2008, and
  12th International Workshop, {RANDOM} 2008, Boston, MA, USA, August 25-27,
  2008. Proceedings}, pages 523--534, 2008.
\newblock URL: \url{https://doi.org/10.1007/978-3-540-85363-3\_41}, \href
  {http://dx.doi.org/10.1007/978-3-540-85363-3\_41}
  {\path{doi:10.1007/978-3-540-85363-3\_41}}.

\bibitem{mahoney2011randomized}
Michael~W Mahoney.
\newblock Randomized algorithms for matrices and data.
\newblock {\em arXiv preprint arXiv:1104.5557}, 2011.

\bibitem{MahoneyD09}
Michael~W. Mahoney and Petros Drineas.
\newblock {CUR} matrix decompositions for improved data analysis.
\newblock {\em Proc. Natl. Acad. Sci. {USA}}, 106(3):697--702, 2009.
\newblock URL: \url{https://doi.org/10.1073/pnas.0803205106}, \href
  {http://dx.doi.org/10.1073/pnas.0803205106}
  {\path{doi:10.1073/pnas.0803205106}}.

\bibitem{Sarlos06}
Tam{\'{a}}s Sarl{\'{o}}s.
\newblock Improved approximation algorithms for large matrices via random
  projections.
\newblock In {\em 47th Annual {IEEE} Symposium on Foundations of Computer
  Science {(FOCS} 2006), 21-24 October 2006, Berkeley, California, USA,
  Proceedings}, pages 143--152, 2006.
\newblock URL: \url{https://doi.org/10.1109/FOCS.2006.37}, \href
  {http://dx.doi.org/10.1109/FOCS.2006.37} {\path{doi:10.1109/FOCS.2006.37}}.

\bibitem{WangZ13}
Shusen Wang and Zhihua Zhang.
\newblock Improving {CUR} matrix decomposition and the nystr{\"{o}}m
  approximation via adaptive sampling.
\newblock {\em J. Mach. Learn. Res.}, 14(1):2729--2769, 2013.
\newblock URL: \url{http://dl.acm.org/citation.cfm?id=2567748}.

\end{thebibliography}
